\documentclass[12 pt]{article}
\usepackage{cite,amsmath,amsfonts,amsthm,fullpage}
\usepackage{youngtab}
\usepackage{cite,amsmath,amsfonts,amsthm,fullpage, amssymb, mathtools, url}
\usepackage[usenames,svgnames]{xcolor}
\usepackage{youngtab, enumitem, comment, mathdots, graphicx, float}
\usepackage[symbol]{footmisc}
\usepackage[pagebackref=true]{hyperref} 

\theoremstyle{plain}
\theoremstyle{definition}
\newtheorem{theorem}{Theorem}[section]
\newtheorem{lemma}[theorem]{Lemma}
\newtheorem{definition-theorem}[theorem]{Definition-Theorem}
\newtheorem{definition-proposition}[theorem]{Definition-Proposition}

\newtheorem{proposition}[theorem]{Proposition}
\newtheorem{corollary}[theorem]{Corollary}
\newtheorem{example}{Example}[section]
\newtheorem{examples}{Example}[subsection]

\newtheorem{remark}{Remark}[section]
\newtheorem{remarks}{Remarks}[section]

\numberwithin{equation}{section} 

\DeclareMathOperator{\End}{End}

\DeclareMathOperator{\GL}{GL}

\DeclareMathOperator{\SO}{SO}

\def\ra{{\rightarrow}}

\def\det{\mathrm {det}}

\def\Pf{\mathrm {Pf}}

\def\End{\mathrm {End}}

\def\span{\mathrm {span}}

\def\span{\mathrm {span}}

\def\&{&{\hskip -20pt}}

\DeclarePairedDelimiter{\no}{:}{:}

\def\be{\begin{equation}}
\def\ee{\end{equation}}

\def\bea{\begin{eqnarray}}
\def\eea{\end{eqnarray}}

\def\bt{\begin{theorem}}
\def\et{\end{theorem}}

\def\bex{\begin{example}\small \rm}
\def\eex{\end{example}}

\def\bexs{\begin{examples}\small \rm}
\def\eexs{\end{examples}}

\def\ra{\rightarrow}
\def\ss{\subset}

\def\br{\begin{remark}\small \rm}
\def\er{\end{remark}}

\def\FF {{\mathcal F}}

\def\HH{{\mathcal H}}

\def\QQ {{\mathcal Q}}

\def\Mb{\mathbf{M}}
\def\Nb{\mathbf{N}}

\def\Zb{\mathbf{Z}}
\def\ab{\mathbf{a}}

\def\pb{\mathbf{p}}

\def\rb{\mathbf{r}}

\def\tb{\mathbf{t}}

\def\xb{\mathbf{x}}
\def\yb{\mathbf{y}}

 \def\grg{\mathfrak{g}}

 \def\grl{\mathfrak{l}}

 \def\gro{\mathfrak{o}}

 \def\grs{\mathfrak{s}}

\def\bp{\begin{Proposition}\rm}
\def\ep{\end{Proposition}}
\def\bc{\begin{corollary}}
\def\ec{\end{corollary}}
\def\bl{\begin{lemma}\em}
\def\el{\end{lemma}}
\def\be{\begin{equation}}
\def\ee{\end{equation}}
\def\br{\begin{remark}\rm\small}
\def\er{\end{remark}}
\def\brs{\begin{remarks}.\\ \rm\
\begin{enumerate}}
\def\ers{\end{enumerate}\end{remarks}}
\def\bea{\begin{eqnarray}}
\def\eea{\end{eqnarray}}

\begin{document}

\begin{center}
\begin{Large}\fontfamily{cmss}
\fontsize{17pt}{27pt}
\selectfont
	\textbf{Polynomial KP and BKP $\tau$-functions and correlators}
	\end{Large}
	
\bigskip \bigskip
\begin{large} 
J. Harnad$^{1, 2}$\footnote[1]{e-mail:harnad@crm.umontreal.ca} 
and A. Yu. Orlov$^{3,4}$\footnote[2]{e-mail:orlovs55@mail.ru} 
 \end{large}
 \\
\bigskip

\begin{small}
$^{1}${\em Centre de recherches math\'ematiques, Universit\'e de Montr\'eal, \\C.~P.~6128, succ. centre ville, Montr\'eal, QC H3C 3J7  Canada}\\
$^{2}${\em Department of Mathematics and Statistics, Concordia University\\ 1455 de Maisonneuve Blvd.~W.~Montreal, QC H3G 1M8  Canada}\\
$^{3}${\em Shirshov Institute of Oceanology, Russian Academy of Science, Nahimovskii Prospekt 36, Moscow 117997, Russia }\\
$^{4}${\em 
ITEP (Institute of Theoretical and Experimental Physics), Moscow, ul. B.Cheremushkinskaya 25, Russia, 117218}
\end{small}
 \end{center}

\medskip

\begin{abstract}
Lattices of polynomial KP and BKP $\tau$-functions labelled by partitions, with  the flow variables equated
 to finite power sums, as well as associated multipair KP  and multipoint  BKP correlation functions,
 are expressed via generalizations of  Jacobi's bialternant formula for Schur functions  and 
Nimmo's Pfaffian ratio formula for Schur $Q$-functions. These are obtained by applying Wick's theorem to 
fermionic  vacuum expectation value representations in which the infinite group element acting 
on the lattice of basis states stabilizes the vacuum.
     \end{abstract}
\bigskip


\section{Lattices of KP and BKP $\tau$-functions: fermionic constructions and polynomial solutions}

In \cite{HO2}  fermionic vacuum expectation value (VEV) representations were used to construct lattices of KP $\tau$-functions 
\cite{Sa} labelled by integer partitions, and BKP $\tau$-functions \cite{DJKM1} labelled by strict partitions.  
When the underlying infinite 
$\GL(\infty)$ and $\SO(\infty)$ group elements are related by an appropriately defined factorization, the KP $\tau$-functions, 
restricted to vanishing  even flow variables, were shown to be expressible as finite sums over products of BKP $\tau$-functions. 
It was also shown that, choosing group elements that stabilize the vacuum,   the resulting $\tau$-functions are 
symmetric polynomials, generalizing the ``building block'' solutions  for the KP  and BKP hierarchies consisting 
of Schur functions $s_\lambda$ and Schur $Q$  functions $Q_\alpha$, respectively,  thereby giving a further perspective 
on the well-studied classes  of KP and BKP $\tau$-functions  of polynomial  
 type \cite{BL1, BL2, You,  SV, Iv,  BFH,  vdLO, KvdL1,  KvdL2, HL, Roz}.
 
 In the present work, it is shown that when the flow variables are restricted to power sums in 
 a finite number of auxiliary variables,  Wick's theorem  implies finite determinantal and Pfaffian representations
 for such polynomial $\tau$-functions, as well as $n$-pair KP correlators  and $n$-point 
 BKP correlators,  similar to Jacobi's bialternant expression for Schur functions and Nimmo's 
 Pfaffian ratio formula for Schur's Q-functions.
 
 Following the approach developed by the Sato school  \cite{Sa, DJKM1, DJKM2, JM},
 the basic framework of operators and  flows on fermionic Fock space is recalled in Section \ref{fermi_VEV_KP_BKP},
 both for the charged fermion case and for a pair of commuting classes of neutral fermions.
Charged fermions are used to construct VEV representations of a lattice of KP $\tau$-functions
 $\{\pi_\lambda(g)(\tb)\}$,  where $\tb = (t_1, t_2, \dots )$ is the infinite sequence of KP flow variables,
as well as $n$-pair correlators, labelled by pairs $(g, \lambda)$ consisting of an element $g\in \GL(\HH)$ of the
 infinite group of  invertible general linear transformations on an underlying Hilbert space $\HH$ and an integer partition
 $\lambda$.  Neutral fermions are similarly used to construct pairs of VEV representations of  a lattice of BKP 
 $\tau$-functions  $\{\kappa_\alpha(h^\pm)(\tb_B)\}$,  where $\tb_B=(t_1, t_3, \dots)$ is the infinite sequence 
 of BKP flow variables,  and $2n$-point correlators, labelled by pairs $(h^\pm , \alpha)$ consisting of an element 
 $h^\pm\in \SO(\HH_{\phi^{\pm}})$ of the infinite group of special orthogonal transformations on an underlying 
pair of mutually orthogonal, complementary subspaces $\HH_{\phi^\pm}\ss \HH + \HH^*$ 
of the direct sum of $\HH$ with its dual $\HH^*$,  with respect to the natural scalar product, 
  and a strict partition $\alpha$,

 It is known that when the group element $g\in \GL(\HH)$ is chosen to be
  upper triangular, and hence its fermionic representation $\hat{g}$ stabilizes the vacuum state  
  $|0\rangle$, the resulting KP $\tau$-functions  are polynomials in the KP flow variable $\tb=(t_1, t_2, \dots)$
  and, in fact, all polynomial KP $\tau$-functions are expressible in this way \cite{KvdL1, KvdLRoz}.
 In Section \ref{generalized_bialtern}, it is shown that when the flow variables are 
  restricted to finite power sums 
 \be
 t_j = [\xb]_j := \tfrac{1}{j} \sum_{a=1}^n x_a^j,
 \label{n_power_sums}
 \ee
 in a set of $n$ auxiliary variables $\xb= (x_1, \dots, x_n)$,  Wick's theorem leads to an expression for
 the KP $\tau$-functions  as the ratio of a finite determinant of alternant form  and the Vandermonde determinant, 
 with the entries in the numerator alternant polynomials in the $\xb$ variables, as  in Jacobi's bialternant fomula 
 for Schur functions \cite{Mac1}. The resulting formula for the  KP $\tau$-function is
 thus a generalization of Jacobi's formula. Although these $\tau$-functions are, in general, inhomogeneous 
 symmetric polynomials, they nevertheless share many of the properties of Schur functions \cite{Mac1, HL, SV},
 such as the Giambelli identity \cite{Mac1}.
  
It is also known,  in the BKP case, that when the group elements $h^\pm\in \SO(\HH_{\phi^\pm})$ 
are similarly chosen to be upper triangular, so the fermionic representation $\hat{h}^\pm$ again
  stabilizes the vacuum state  $|0\rangle$, the resulting BKP $\tau$-functions are again polynomials in the BKP
  flow variables $\tb_B=(t_1, t_3, \dots )$ and, in fact, this exhausts the full set of polynomial BKP $\tau$-functions \cite{KvdL2}.
In Section \ref{BKP_gen_nimmo}, it is shown that if the flow variables are  restricted to power sums $[\xb]$ in the auxiliary variables 
$(x_1, \dots, x_n)$,  Wick's theorem leads to an expression for the BKP $\tau$-function as a ratio,  in which the numerator is a 
finite Pfaffian whose entries are either polynomials in an even number  of auxiliary variables $\{x_a\}_{a=1, \dots 2n} $'s,  or  
 rational combinations $\{M_{ab} :=\frac{x_a-x_b }{ x_a + x_b}\}$  of these,  and the denominator is the Pfaffian 
 of the skew  $2n \times 2n$ matrix $M$ with only these rational entries,  thereby generalizing  Nimmo's 
 formula \cite{Nim} for Schur  $Q$ functions, with which these, again, share many properties.
  
  In Section \ref{KP_BKP_correls},  analogous finite determinantal formulae are deduced for the $n$-pair  correlation functions
  associated to the lattice  of KP $\tau$-functions $\{\pi_\lambda(g)(\tb)\}$. For group elements $\hat{g}$ that
  stabilize the vacuum, these again reduce to polynomials in the coordinates of the points $\{x_a, y_a\}_{a=1, \dots, n} $
  appearing in the correlation function, multiplied by an explicit rational factor.
   For the BKP case, the $2n$-point correlation functions  are similarly expressed 
in terms of restrictions of  the BKP $\tau$-functions $\{\kappa_\alpha(h^\pm)(\tb_B)\}$ to power sum 
variables in the coordinates of the points. When the group elements $\hat{h}^\pm$ stabilize
  the vacuum, the generalized Nimmo formula derived in Section  \ref{BKP_gen_nimmo} thus
  provides a Pfaffian ratio expression for the $2n$-point correlators as polynomials in the point coordinates
 multiplied by explicit rational factors.
  
  Section \ref{examples} illustrates these results with some examples of polynomial KP and BKP $\tau$-functions,
  as well as with KP  and BKP $2$-point correlators.

\section{Fermionic VEV representations of lattices of KP and BKP $\tau$-functions}
\label{fermi_VEV_KP_BKP}

We begin by recalling the construction \cite{HO2} of lattices of KP  $\tau$-functions $\{\pi_\lambda(g)(\tb)\}$
and BKP  $\tau$-functions $\{\kappa_\alpha(h) (\tb_B)\}$ labelled, respectively, by pairs $(g, \lambda)$ consisting,
in the KP case, of an infinite group element $g\in \GL(\HH)$) and  an integer partition 
 $\lambda=(\lambda_1, \dots, \lambda_{\ell(\lambda)})$,  and in the BKP case by pairs $(h^\pm, \alpha)$ 
 consisting of  an infinite orthogonal  group element $h^\pm \in \SO(\HH_{\phi^\pm})$ and a strict partition 
 $\alpha =(\alpha_1, \dots, \alpha_r)$  with an even number $r$ of parts.
(A survey of the use of fermionic methods in the theory of $\tau$-functions may be found in \cite{MJD}, Chapt.~3
and \cite{HB}, Chapts.~5 and 7.)
 
 \subsection{Lattice of KP $\tau$-functions $\pi_\lambda(g) (\tb) $}
 \label{KP_pi_lambda}
 
 The lattice of KP $\tau$-functions introduced   in \cite{HO2} may be expressed as 
 fermionic vacuum state expectation values,
  \be
 \pi_\lambda(g) (\tb) :=  \langle 0 | \hat{\gamma}_+(\tb) \hat{g}) | \lambda \rangle,
 \label{pi_lambda_A_def}
 \ee
 where $|0\rangle$ is the vacuum state in the charge  zero  sector of the fermionic Fock space $\FF$, 
 which is the semi-infinite exterior product space on a separable Hilbert space $\HH$ with denumerable basis
 $\{e_j\}_{j\in\Zb}$
\be
 \FF:= \Lambda^{\infty/2} =\bigoplus_{n\in \Zb} \FF_n,
 \ee
 $|\lambda\rangle$ is an orthonormal basis element in  the zero fermionic charge sector 
 $\FF_0$, labelled  by an integer partition $\lambda$,  
and
  \be
 \tb = (t_1, t_2, \dots )
 \ee
 denotes the KP flow variables. The latter may alternatively be interpreted as the evaluation
 of normalized power sums
 \be
 \tb = [\xb] := ([\xb ]_1, [\xb]_2, \dots, [\xb ]_j, \dots),  
  \label{power_sum_tb}
  \ee
  where
  \be
  [\xb]_j := \tfrac{1}{j}\sum_{a=1}^n x_a^j, \quad j\in \Nb^+
 \ee
  in terms of a finite or infinite set of bosonic variables $\xb = (x_1, \dots , x_n)$
  
 The dual basis elements $\{e^j \}_{j\in \Zb}$ for  $\HH^*$ are defined by
\be
e^j(e_k) = \delta_{jk}, \quad j,k \in \Zb
\ee
and the vacuum state in the $\FF_0$ sector is denoted
 \be
 |\emptyset;0\rangle = |0\rangle := e_{-1} \wedge e_{-2} \wedge \cdots.
 \ee
The  charged fermionic creation and annihiliation operators, $\{\psi_, \psi^\dag_j \in \End(\FF)\}_{j\in \Zb}$
are defined as the exterior and inner products with the basis elements, and dual basis elements, respectively
 \be
 \psi_j := e_j \wedge, \quad \psi^\dag_j := i_{e^j}.
 \ee
They satisfy the usual anti-commutation relations
 \be
 [\psi_j, \psi^\dag_k]_+ = \delta_{jk}, \quad  [\psi_j, \psi_k]_+ = 0, \quad  [\psi^\dag_j, \psi^\dag_k]_+ =0,    \quad j,k \in \Zb
 \label{psi_psi_dag_anticomm} 
 \ee
and generate the fermionic representation of the Clifford algebra on $\HH + \HH^*$ corresponding to the scalar product
 \be
 Q(v + \mu, w + \nu):= \nu(v) + \mu(w), \quad v,w \in \HH, \ \mu, \nu \in \HH^*.
 \label{Q_def}
 \ee
They also satisfy the vacuum annihilation conditions
 \bea
 \psi_{-j} |0 \rangle &\& =0, \quad \psi^\dag_{j-1}|0\rangle = 0,  
 \label{charged_fermi_vac_annih}
 \\
 \langle 0 |\psi^\dag_{-j}  &\& =0, \quad \langle 0 | \psi_{j-1} = 0, \quad j\in \Nb^+  .
 \label{charged_fermi_dual_vac_annih}
 \eea 
  
 For an integer partition  $\lambda = (\lambda_1, \dots ,\lambda_{\ell(\lambda)})$ of length $\ell(\lambda)$
  with Frobenius  indices \cite{Mac1},   
 \be
 (\alpha|\beta) = (\alpha_1, \dots, \alpha_r|\beta_1, \dots, \beta_r),
 \label{alpha_beta}
 \ee
 the basis state $|\lambda\rangle$ in the $n=0$ sector  $\FF_0$ is  ({\hskip -.5 pt}\cite{JM},  Chapt. 3, \cite{HB},  Chapt.~5)
 \be
|\lambda\rangle := |\lambda; 0\rangle= (-1)^{\sum_{j=1}^r \beta_j }\prod_{j=1}^r \psi_{\alpha_j} \psi^\dag_{\-\beta_j -1} | 0\rangle.
\label{lambda_0_bsasis_state}
 \ee
 When substituted in  (\ref{pi_lambda_A_def}), this yields the VEV representation
   \be
 \pi_\lambda(g) (\tb) := 
 (-1)^{\sum_{j=1}^r \beta_j }\ \langle 0 | \hat{\gamma}_+(\tb) \hat{g} \prod_{j=1}^r \psi_{\alpha_j} \psi^\dag_{\-\beta_j -1} | 0\rangle,
 \label{pi_lambda_A_VEV_re[}
 \ee
 
 More generally,  an orthonormal basis $\{|\lambda; n\rangle\}$ for the charge $n$ subspace $\FF_n\ss \FF$
is provided by
\bea
|\lambda; n\rangle  &\& = (-1)^{\sum_{j=1}^r \beta_j }\prod_{j=1}^r \psi_{\alpha_j+n} \psi^\dag_{\-\beta_j +n-1} |n\rangle
\label{psi_alpha_beta_vac_n}\\
&\& = e_{l_1(n)} \wedge e_{l_2(n)} \wedge \cdots \cr
 &\&= \psi_{l_1(n)} \wedge \cdots \wedge \psi_{l_{\ell(\lambda)}(n)} | n - \ell(\lambda)\rangle, \quad n \in \Zb, 
 \label{lambda_n_basis}
\eea
where
\be
l_i (n):= \lambda_i - i + n, \quad i\in \Nb^+
\ee
are the {\em particle positions} (({\hskip -.5 pt}\cite{HB}, Chapt.~5)   associated with the partition $\lambda$ in the $\FF_n$ sector
(with $\lambda_j:= 0$ for $j> \ell(\lambda)$),  and
\be
|n \rangle := |\emptyset; n\rangle = e_{n-1} \wedge e_{n-2}\wedge \cdots
\ee
is the  vacuum state in the $\FF_n$ sector.  The dual basis vectors, denoted $\{\langle \lambda; n|\}$, satisfy
\be
\langle \lambda; n| \mu ; m \rangle = \delta_{\lambda \mu} \delta_{nm}.
\ee
  
 The fermionic representation $\hat{g}$ of group elements $g \in \GL_0(\HH)$ in the 
 identity component is
 \be
\hat{g}:=\hat{g}(\tilde{A}) = e^{\hat{\tilde{A}}},
\label{hat_g_exp_A_def}
\ee
where
\be
\hat{\tilde{A}} :=\sum_{i,j \in \Zb} \tilde{A}_{jk} \no{\psi_j \psi^\dag_k}, 
\label{hat_tilde_A_def}
\ee
with
$\{\tilde{A}_{jk}\}_{j,k \in \Zb}$ the elements of a doubly infinite matrix $\tilde{A}$ such
that, in the basis $\{e_j\}_{j\in \Zb}$  $g$ is represented by
\be
g (\tilde{A})= e^{\tilde{A}}.
\label{g_tilde_A}
\ee
Normal ordering of the product of  a pair of linear elements
\be
\no{\hat{L}_1 \hat{L}_2} = \hat{L}_1 \hat{L}_2 - \langle 0 | \hat{L}_1 \hat{L}_2 | 0 \rangle
\label{normal_order}
\ee 
is defined so the vacuum expectation value (VEV) vanishes.

 The  KP  flows are generated by the infinite abelian subgroup $\Gamma_+\ss \GL_0(\HH)$
of shift  flows
\be
\Gamma_+ = \{\gamma_+(\tb) := e^{\sum_{j=1}^\infty t_j\Lambda^j}, \quad \Lambda (e_i) = e_{i-1}, \ i \in \Zb
\ee
  whose elements $\gamma_+(\tb)$ are represented fermionically as
 \be
  \hat{\gamma}_+(\tb) := e^{\sum_{j=1}^\infty t_j J_j},
 \ee
where
\be
J_j := \sum_{k\in \Zb} \psi_k \psi^\dag_{k+j}, \quad  j  \in \Nb^+
\label{J_j_def}
\ee
are the charged current components. These mutually commute
\be
[J_j, J_k] =0, \quad \forall \ j,k \in \Nb^+
\ee
and annihilate the vacuum state
\be
J_j |0 \rangle = 0, \quad \forall \ j \in \Nb^+.
\ee

The lattice of KP $\tau$-functions (\ref{pi_lambda_A_def}) can be extended to an infinite sequence of KP $\tau$-functions
\bea
\pi_{\lambda, n} (g) (\tb) &\&:=  \langle n | \hat{\gamma}_+(\tb) \hat{g} | \lambda ; n \rangle =(-1)^{\sum_{j=1}^r \beta_j }\ \langle n | \hat{\gamma}_+(\tb) \hat{g}) \prod_{j=1}^r \psi_{\alpha_j} \psi^\dag_{\-\beta_j -1} | n\rangle,
 \label{pi_lambda_n_A_def}
\eea
defined in each sector $\FF_n$ which, for each pair $(g, \lambda)$, form an integer
 lattice of mKP $\tau$-functions. (See  \cite{JM, DJKM2} or \cite{HB},  Chapt.~7.)

We also introduce fermionic field operators
\be
\psi(z) := \sum_{j\in \Zb} \psi_j z^j, \quad \psi^\dag_j(z) := \sum_{j\in \Zb} \psi^\dag_j z^{-j-1}
\label{psi_fields}
\ee
and ``dressed'' creation and annihilation operators
\bea
\psi_j(\tilde{A})&\& :=\hat{g}(\tilde{A})\psi_j  \hat{g}^{-1}(\tilde{A}) = \sum_{k\in \Zb} g_{kj}(\tilde{A}) \psi_k
\label{psi_tilde_A_def} \\
  \psi_j^*(\tilde{A})
&\&:=\hat{g}(\tilde{A})\psi_j^\dag  \hat{g}^{-1}(\tilde{A}) =\sum_{k\in \Zb} g^{-1}_{jk}(\tilde{A}) \psi^\dag_k. 
\label{psi_star_tilde_A_def}
\eea

In the special case where the matrix $\tilde{A}$ is strictly upper triangular 
\be
\tilde{A}_{ij} = 0 \ \text{ if } i\ge j
\label{tilde_A_triang}
\ee
the $\grg\grl(\HH)$ algebra elements $\hat{\tilde{A}}$ annihilate the vacuum state
\be
\hat{\tilde{A}}|0 \rangle = 0, \quad \hat{A}^\pm|0\rangle,
\label{tilde_A_triang_vac_annihil}
\ee
and the group elements $\hat{g}(\tilde{A})$  stabilize it
\be
\hat{g}(\tilde{A})|0 \rangle.
\label{g_vac_stabil}
\ee
We then have
\bea
\psi_j(\tilde{A})&\&  = \sum_{k=-\infty}^j g_{kj}(\tilde{A}) \psi_k, \quad g_{kk}(\tilde{A})  =1, \ \forall  k \in \Zb, 
\label{psi_tilde_A_triang} \\
  \psi_j^*(\tilde{A}) &\& =\sum_{k=j} ^\infty g^{-1}_{jk}(\tilde{A}) \psi^\dag_k, \quad g^{-1}_{kk}(\tilde{A})  =1, \ \forall  k\ \in \Zb.
\label{psi_star_tilde_A_triang}
\eea
Now define two sequences of monic polynomials
\bea
p_j(x|\tilde{A})&\&:= \sum_{k=0}^j P_{kj}(\tilde{A}) x^k, \quad P_{jj}(\tilde{A}) = 1, \quad j\in \Nb,
\label{p_j_A_tilde_x}
 \\
p^*_j(y|\tilde{A})&\&:= \sum_{k=0}^j P^*_{jk}(\tilde{A}) y^k, \quad P^*_{jj}(\tilde{A}) = 1, \quad j\in \Nb,
\label{p_star_j_A_tilde_x}
\eea
where the upper triangular matrix of coefficients $\{P_{jk}(A)\}_{j,k \in \Nb}$ is
the $\Nb \times \Nb$ block of the $\Zb \times \Zb$ upper triangular matrix $g(\tilde{A})$ 
obtained by exponentiating the strictly upper triangular matrix $\tilde{A}$ as in (\ref{g_tilde_A})
\be
P_{jk}(\tilde{A}) = g_{jk}(\tilde{A}),   \quad P_{jj}=1, \quad j\le k, \ j,k \in \Nb.
\label{P_kj_coeffs}
\ee
while the lower triangular matrix of coefficients $\{P^*_{jk}(A)\}_{j,k \in \Nb}$ is the $(-\Nb^+ )\times (-\Nb^+)$
block of the inverse matrix $g^{-1}$
\be
P^*_{jk}= g^{-1}_{-j-1,-k-1}.  \quad P^*_{jj}=1, \quad j\ge k, \ j,k \in \Nb.
\label{P_star_jk_coeffs}
\ee

\begin{lemma}
\label{VEV_p_j}
The monic polynomials sequences $\{p_j(x|\tilde{A}) \}_{j\in \Nb}$  and $\{p^*_j(y|\tilde{A}) \}_{j\in \Nb}$
have the following VEV representations:
\bea
p_j(x | \tilde{A}) &\&=  x^{-1}\langle 0 | \psi^\dag(x^{-1}) \psi_j(\tilde{A})|0\rangle \quad \forall \, j  \in \Nb,
\label{A_tilde_polynoms} \\
p^*_j(y | \tilde{A}) &\&=  y^{-1}\langle 0 | \psi(y^{-1}) \psi^*_{-j-1}(\tilde{A})|0\rangle \quad \forall \, j  \in \Nb,
\label{A_tilde_polynoms_dual} 
\eea
\end{lemma}
\begin{proof} Substitute expressions (\ref{psi_tilde_A_triang}) for  $\psi_j(\tilde{A})$ 
and (\ref{psi_star_tilde_A_triang}) for $\psi^*_j(\tilde{A})$, and the series expansions
(\ref{psi_fields}) for $\psi^\dag(x^{-1})$ and $\psi(y^{-1})$  into $\langle 0 | \psi^\dag(x^{-1}) \psi_j(\tilde{A})|0\rangle$,
and  $\langle 0 | \psi(y^{-1}) \psi^\dag_j(\tilde{A})|0\rangle$ and evaluate the terms in the sum using
(\ref{charged_fermi_vac_annih}), (\ref{charged_fermi_dual_vac_annih}), which imply
 \bea
\langle 0 | \psi^\dag(x^{-1}) \psi_j|0 \rangle &\&= x^{j+1},  \ \text{ for } j\ge 0  \\
\langle 0 | \psi(y^{-1}) \psi^\dag_{-j-1}|0 \rangle &\&= y^{j+1},  \ \text{ for } j\ge 0 .
\eea
\end{proof}

It follows \cite{HO2, KvdL1} that the associated KP $\tau$-functions $\pi_\lambda(g(\tilde{A})) (\tb)$  defined in (\ref{pi_lambda_A_def}) are polynomials in the flow variables $\tb$. Besides polynomiality,  they share many further properties 
with Schur functions $s_\lambda(\tb)$, so if conditions (\ref{tilde_A_triang})-(\ref{g_vac_stabil}) are satisfied, we denote these as
\be
s_\lambda(\tb | \tilde{A}) := \pi_\lambda(g(\tilde{A)}) (\tb).
\label{s_lambda_A_def}
\ee
In particular, choosing  $\tilde{A}$ to vanish, so that  $g(\tilde{A})$  is the identity element,
 we recover the Schur functions
\be
 s_\lambda(\tb | 0) = s_\lambda(\tb).
\ee
More generally, if $\tilde{A}$ is upper triangular, we can extend the definition (\ref{s_lambda_A_def}) of polynomial KP $\tau$-functions to an integer lattice of polynomial mKP $\tau$-functions for each $\lambda$
\be
s_{\lambda, n} (\tb | \tilde{A}) := \pi_{\lambda, n}(g(\tilde{A)}) (\tb),  \quad n \in \Zb.
\label{s_lambda_n_A_def}
\ee


 \subsection{Lattice of BKP $\tau$-functions $\kappa_\alpha(h) (\tb_B)$}
 \label{BKP_kappa_alpha}
  As in the KP case, a lattice of BKP $\tau$-functions was introduced in \cite{HO2},  which may also be expressed as 
 fermionic VEV's,  but with the charged fermionic operators $\{\psi_j, \psi^\dag_j\}_{j\in \Zb}$ 
 replaced by either of a pair of sequences  $\{\phi^+_j\}_{j\in \Zb}$, $\{\phi^-_j\}_{j\in \Zb}$  of  mutually  anti-commuting neutral 
 fermionic operators $\{\phi_j^\pm\}_{j\in \Zb}$,  defined, as in  \cite{DJKM1, JM, You},  by 
 \bea
 \phi^+_j &\&:= \tfrac{1}{\sqrt{2} }(\psi_j + (-1)^j \psi^\dag_{-j}), \\
 \phi^-_j &\&:= \tfrac{i}{\sqrt{2}} (\psi_j - (-1)^j \psi^\dag_{-j}).
 \label{phi_pm_def}
 \eea
These satisfy the anti-commutation relations 
  \be
 [\phi^+_j, \phi^+_k]_+ = (-1)^j \delta_{j, -k}, \quad [\phi^-_j, \phi^-_k]_+ = (-1)^j \delta_{j, -k}, \quad [\phi^+_j, \phi^-_k]_+ = 0, \quad j,k \in \Zb  
   \label{phi_pm_anticomm}
 \ee
 and vacuum annihilation conditions
 \be
 \phi^\pm_{-j} |0\rangle =0,  \quad \langle 0| \phi^\pm_{j}  =0, \quad j>0.
 \label{neutral_fermi_vac_annih}
 \ee
Their pairwise expectation values are:
\bea
    \langle 0| \phi^+_j \phi^+_k |0\rangle&\& =\langle 0| \phi^-_j\phi^-_k|0\rangle =
    \begin{cases}
      (-1)^k\delta_{j,-k}& \text{if}\ k>0,\\
      \tfrac12\delta_{j,0}& \text{if}\ k=0,\\
      0& \text{if}\ k<0,
    \end{cases}
    \label{phi_pairing}
     \\
\langle 0|\phi^+_j\phi^-_k|0\rangle &\&= -\langle 0|\phi^-_j \phi^+_k|0\rangle =\tfrac {i}{2} \delta_{j,0}\delta_{k,0}
\label{phi_pm_pairing}.
\eea

The direct sum $\HH + \HH^*$ decomposes  into an orthogonal direct sum with respect to the 
scalar product $Q$ defined in (\ref{Q_def}),
\be
\HH= \HH_{\phi^+} \oplus \HH_{\phi^-}
\ee
of two subspaces
\be
\HH_{\phi^+} = \span\{f^+_j\}_{j\in \Zb}, \quad \HH_{\phi^-} = \span\{f^-_j\}_{j\in \Zb},
\ee
where the bases  $\{f^+_j\}_{j\in \Zb}$ and $\{f^-_j\}_{j\in \Zb}$, defined by 
\be
f^+_j := \tfrac{1}{\sqrt{2}} (e_j + (-1)^j e^{-j}), \quad f^-_j := \tfrac{i}{\sqrt{2} }(e_j - (-1)^j e^{-j}),
\ee
satisfy the orthogonality relations
\be
Q_\pm(f^\pm_j, f^\pm _k) = (-1)^k\delta_{j+k,0},  \quad \forall \ j,k \in \Zb
\ee
with respect to the scalar products
\be
Q_\pm := Q|_{\HH_{\phi^\pm}}
\ee
on the subspaces $\HH_{\phi^\pm}$ obtained  by restriction of $Q$.

The elements  $h^\pm(A) \in \SO(\HH_{\phi^\pm})$ of the corresponding mutually commuting orthogonal subgroups 
$\SO(\HH_{\phi^\pm}) \ss \SO(\HH+\HH^*, Q)$ have fermionic representations 
 \be
\hat{h}^\pm:=\hat{h}^\pm(A) := e^{\hat{A}^\pm} ,
\label{hat_h_exp_A_def}
\ee
that leave invariant the respective subspaces $\FF_{\phi^\pm}\ss\FF$,
where
\be
\hat{A}^\pm := \tfrac{1}{2}\sum_{j,k \in \Zb}  A_{jk} \no{\phi^\pm_j \phi^\pm_k},
\label{hat_A_def}
\ee
with
 $\{A_{jk}\}_{j,k \in \Zb}$ the elements of a doubly infinite  skew symmetric matrix $A$ determining 
 the matrix representation $h$ of $h^\pm$ in the bases $\{f^\pm_j\}_{j\in \Zb}$ by
\be
h(A) = e^{\check{A}}
\label{h_check_exp_A}
\ee
where
\be
\check{A}_{jk} := (-1)^k A_{j, -k}.
\label{check_A_def}
\ee

 The  BKP flows are generated by infinite abelian subgroups 
 $\Gamma^{B\pm}\ss \SO(\HH_{\phi^\pm})$ whose elements are represented fermionically as
 \be
  \hat{\gamma}^{B \pm}(\tb_B) := e^{\sum_{j=1}^\infty t_{2j-1}J^{B\pm}_j}, 
 \ee
where the neutral current components $\{J^{B\pm}_j\}_{j\in \Nb^+}$ are defined as
\be
J^{B+}_j:=\tfrac 12 \sum_{k\in\mathbb{Z}} (-1)^{k+1}\phi^+_k \phi^+_{-k-j},\quad 
J^{B-}_j=\tfrac 12 \sum_{k\in\mathbb{Z}} (-1)^{k+1}\phi^-_k\phi^-_{-k-j}, \quad j\in \Nb^+.
\label{J_j_B_def}
\ee
Of these,  the even ones $J^{B\pm}_{2j}$ vanish, while the odd ones mutually commute:
\be
[J^{B+}_{2j-1}, J^{B+}_{2k-1}]  = 0,\quad [J^{B-}_{2j-1}, J^{B-}_{2k-1}]=0,\quad
[J^{B+}_{2j-1}, J^{B-}_{2k-1}] =0, \quad j,k \in \Nb^+ 
\label{JB_pm_crs}
\ee
and annihilate the vacuum state
\be
J^{B\pm}_{2j-1}|0\rangle =0, \quad \forall \  j \in \Nb^+.
\ee
They are related to the odd charged current components by
\be
 J_{2j-1}=J^{B+}_{2j-1}+J^{B-}_{2j-1}, \quad \forall \, j \in \Nb^+.
 \label{J_J_Bpm}
\ee
Following \cite{DJKM1, DJKM2, JM, You}, there are two types of neutral fermionic 
basis states
\be
 |\alpha^\pm ) := \phi^\pm_{\alpha_1} \cdots \phi^\pm_{\alpha_r} |0\rangle,
 \label{alpha_pm_state}
 \ee
  spanning  two subspaces  $\FF_{\phi^\pm} \ss \FF$  
 \be
 \FF_{\phi^\pm} := \span\{|\alpha^\pm)\}.
 \label{FF_phi_pm_def}\\
 \ee
 
 \begin{remark}
 Note that the subspaces $\FF_{\phi^\pm}\ss \FF$  are not mutually orthogonal.  In fact, their intersection is infinite dimensional, 
 as is their  intersection with each of the fermionic charge sectors $\FF_n$. However they are invariant, respectively, under the
 two different infinite, mutually commuting subgroups $\SO(\HH_{\phi^\pm}) \ss \SO(\HH+\HH^*, Q)$  
 represented fermionically  by the elements $\{\hat{h}^\pm\}$.
  Since the two subgroups $\SO(\HH_{\phi^\pm})$ are isomorphic, as are their abelian subgroups 
 $\Gamma^{B\pm}$ it is sufficient, in studying the resulting BKP $\tau$ functions, to consider only one of them.  
 However, for consistency with earlier work \cite{HO1, HO2}, in which these were related 
 bilinearly to the corresponding   lattices of KP $\tau$-functions (\ref{pi_lambda_A_def}), we retain 
 here the notation for both types of  operators $\{\phi_j^\pm\}_{j\in \Zb}$ and fermionic Fock spaces  
 $\FF_{\phi^\pm} \ss \FF$,   although there is no difference in the resulting BKP $\tau$-functions or correlators constructed from them.
 \end{remark}
  
The lattice of BKP $\tau$-functions $\{\kappa_\alpha(h)(\tb_B )\}$  is defined  \cite{HO2} as :
\be
\kappa_\alpha(h)(\tb_B ) := \langle 0 | \hat{\gamma}^{B\pm}(\tb_B)  \hat{h}^\pm | \alpha^\pm ) ,
  \label{kappa_alpha_A_def}
  \ee
where
  \be
 \tb_B= (t_1, t_3, \dots)
 \ee
 denote the  BKP flow variables,  which may be restricted, as in (\ref{power_sum_tb}), to evaluations
 on normalized power sums in an auxiliary (finite or infinite) set of bosonic variables 
 \hbox{$\xb=(x_1, x_2, \dots , x_n)$}
  \be
\tb_B = [\xb]_B,\quad 
  [\xb]_B := ([\xb]_1, [\xb]_3, \dots, [\xb]_{2j-1}, \dots),
  \label{power_sum_tb_B}
  \ee

We also introduce fermionic field operators
\be
\phi^\pm(z)  :=\sum_{j\in \Zb}\phi^\pm_j z^j, 
\label{phi_pm_fields}
\ee
which are related to those defined in (\ref{psi_fields}) by
\bea
\psi(z)=\tfrac{1}{\sqrt{2}}\left(\phi^+(z)-i\phi^-(z)  \right), &\&\quad 
\psi^\dag (\tfrac{1}{z})=\tfrac{1}{\sqrt{2}z}\left(\phi^+(-z)+ i\phi^-(-z)  \right)
\label{psi_z_phi_pm_z}
\\
\psi^\dag(-z)\psi(z) &\&=\frac{i}{z}\phi^{+}(z)\phi^{-}(z),
\label{reduction-left}
\eea
and  the ``dressed'' operators
\be
\phi^\pm_j(A):=\hat{h}^\pm(A)\phi_j^\pm (\hat{h}^\pm)^{-1}(A) = \sum_{k\in \Zb} h_{kj}(A) \phi^\pm_k 
\label{phi_pm_A_def}
\ee

In the special case where the matrix $A$ satisfies the antidiagonal triangularity condition
\be
A_{jk} =0  \quad \text{if } j+k \ge 0,
\label{A_antidiag_triang}
\ee
and $\check{A}$ the strictly upper triangular one
\be
\check{A}_{jk} =0  \quad \text{if } j \ge k ,
\label{A_check_triang}
\ee
the $\grs\gro(\HH_{\phi^\pm})$ algebra elements  $\hat{A}^\pm$ annihilate the vacuum state
\be
 \hat{A}^\pm|0 \rangle = 0,
\label{A_triang_vac_annihil}
\ee
and the group elements  $\hat{h}^\pm(A)$  stabilize it
\be
 \hat{h}^\pm(A)|0 \rangle = |0 \rangle.
\label{h_vac_stabil}
\ee
 The matrix $h(A)$ becomes upper triangular, with $1$'s on the diagonal and (\ref{phi_pm_A_def}) therefore reduces to
\bea
\phi^\pm_j(A)&\&:=\hat{h}^\pm(A)\phi_j^+ (\hat{h}^\pm)^{-1}(A) = \sum_{k=-\infty}^j h_{kj}(A) \phi^\pm_k, 
\label{phi_pm_A_triang}
\\
 h_{jj}(A)&\&=1, \ \forall \   j \in\Zb.
\eea

Defining the upper triangular $\Zb \times \Zb$ matrix $P(A)$ with elements
\be
(P(A))_{jk} := \big(h(A)\big)_{jk} -\tfrac{1}{2}\delta_{k0} \big(h(A)\big)_{0j},
\label{P_A_ij-def}
\ee
we again define a sequence of monic polynomials $\{ p_j(x|A) \}_{j\in \Nb^+}$ ,
\be
  p_j(x|A)  = \sum_{k=0}^j P_{kj}(A) x^k, \quad P_{jj} =1, \ \text{ if } j \neq 0, \ P_{00} = \tfrac{1}{2},
  \label{p_j_tilde_A_def}
\ee
and
\be
 p_0(x|A) =\tfrac{1}{2},
\ee
such that  the upper triangular matrix of coefficients $\{P_{jk}(A)\}_{j,k,\in \Nb}$ is 
the $\Nb \times \Nb$ block $P(A)$. We then have
\begin{lemma}
\label{VEV_p_j_A}
The polynomials $\{p_j(x|A) \}_{j\in \Nb}$  have the following VEV representations:
\bea
p_j(x | A) &\&= \langle 0 | \phi^\pm(-x^{-1}) \phi^\pm_j(A)|0\rangle,  \quad \forall \, j  \in \Nb.
\label{A_polynoms} 
\eea
\end{lemma}
\begin{proof} This follows from
substituting expression (\ref{phi_pm_A_triang}) for  $\phi^\pm_j(A)$ 
and the series expansion
(\ref{phi_pm_fields}) for $\phi^\pm(-x^{-1})$   into $\langle 0 | \phi^\pm(x^{-1}) \phi^\pm_j(A)|0\rangle$,
 and evaluating the terms in the sum using
(\ref{phi_pairing}), which implies
 \be
\langle 0 | \phi^\pm(-x^{-1}) \phi^\pm_j|0 \rangle= \begin{cases} x^j  \,   \text{ for } j> 0 , \cr
  \tfrac{1}{2} \,   \text{ for } j= 0 ,\cr
  0  \,   \text{ for } j<0  .
  \end{cases}
\ee
\end{proof}
It follows \cite{HO2, KvdL2} that the associated BKP $\tau$-functions $ \kappa_\alpha(h(A))(\tb_B )$
are polynomials in the flow variables  $\tb_B$ which, besides polynomiality,  share
many properties with the (scaled) Schur $Q$-functions, which correspond to choosing $A=0$.
Therefore, if conditions 
(\ref{A_antidiag_triang})-(\ref{h_vac_stabil}) are satisfied, we denote these as
\be
\QQ_\alpha([\xb]_B | A):= \kappa_\alpha(h(A) ([\xb]_B),
\label{Q_alpha_A_def}
\ee
and similarly define $Q_\alpha(\xb|A)$ by
\be
Q_\alpha(\xb|A) =  2^{\tfrac{r}{2}} \QQ_\alpha(2[\xb]_B| A) . 
\label{Q_A_QQ_A}
\ee

In particular, choosing $A$ to vanish, so $h(A)$ is the identity element, 
 we recover the Schur Q-functions
\be
Q_\alpha(\xb)  =   2^{\tfrac{r}{2}} \QQ_\alpha( [2\xb]_B | 0).
\ee

\section{Bialternant formula for polynomial KP tau functions}
\label{generalized_bialtern}

Setting $\tb =[\xb]$, as in (\ref{power_sum_tb}), with a finite number $n$ of variables $\xb=(x_1, \dots, x_n)$,
 and assuming the length $\ell(\lambda)$ of the partition $\lambda$ satisfies  $\ell(\lambda) \le n$,  
 Jacobi's bialternant formula for Schur functions \cite{Mac1} is
\be
s_\lambda([\xb]) = \frac{\det\left(x_j ^{\lambda_k - k + n}\right)_{1\le \, j,k \,  \le n}} {\Delta(\xb)},
\ee
where we set $\lambda_k =0$ for $k > \ell(\lambda)$,  and
\be
\Delta(\xb) = \prod_{1\le j < k \le n} (x_j - x_k)
\ee
is the Vandermonde determinant.
 
This can be generalized  by replacing the monomials $\{x^j\}_{j\in \Nb}$ by an arbitrary sequence
of monic polynomials, as in \cite{HL}
\be
p_j(x|\tilde{A}):= \sum_{k=0}^j P_{kj}(\tilde{A}) x^k, \quad P_{jj}(\tilde{A}) = 1, \quad j\in \Nb,
\ee
where the upper triangular matrix of coefficients $\{P_{jk}(A)\}_{j,k \in \Nb}$ is
the $\Nb \times \Nb$ block of the $\Zb \times \Zb$ upper triangular matrix of coefficients 
$P(\tilde{A})$ obtained by exponentiating the strictly upper triangular matrix $\tilde{A}$ as in (\ref{g_tilde_A})
\be
P(\tilde{A}) = g(\tilde{A}) = e^{\tilde{A}}.
\ee
Defining
\be
\tilde{s}_{\lambda, n}([\xb] | \tilde{A}) :=  \frac{\det\left(p_{\lambda_k - k + n} (x_j|\tilde{A})\right)_{1\le\, j,k \, \le n}} {\Delta(\xb)},
\label{s_lambda_tilde_A}
\ee
it  follows that this coincides with the KP $\tau$-function $s_{\lambda, n}(\xb | \tilde{A})$
defined in  (\ref{s_lambda_n_A_def}).
\begin{proposition}
\label{s_lambda_tilde_n}
\be
\tilde{s}_{\lambda, n}([\xb] | \tilde{A}) = s_{\lambda, n}([\xb] | \tilde{A}).
\label{s_lambda_tilde_n_tilde_s}
\ee
\end{proposition}
Although this result was proved in  \cite{HL}, and special cases have long been  studied \cite{BL1, BL2, Ol, SV}, 
we provide here,  for completeness,  a self-contained proof.

\begin{proof}
Recall the following formula, related to the bosonization map  \cite{JM, DJKM1}:
\be
\langle 0 |  \psi^\dag(x_n^{-1}) \cdots \psi^\dag(x_1^{-1})= \Big(\prod_{j=1}^n x_j\Big) \Delta(\xb) \langle n | \hat{\gamma}_+(\sum_{a=1}^n [x_a]) .
\label{n_bosoniz_psi}
\ee
Eq.~(\ref{s_lambda_tilde_n_tilde_s}) is obtained by substituting \ref{n_bosoniz_psi}) in (\ref{pi_lambda_n_A_def}) 
and (\ref{s_lambda_n_A_def}), and choosing $\tb = [\xb]$. 
For upper triangular $\tilde{A}$ and partition $\lambda$,  we have,  from the expression (\ref{lambda_n_basis}) for
the basis element $|\lambda; n\rangle$, the definition  eq.~(\ref{psi_tilde_A_def}) of $\psi_j(\tilde{A})$
 and the fact that $\hat{g}(\tilde{A})$ stabilizes the vacuum (\ref{g_vac_stabil}),
\bea
s_{\lambda, n}([\xb])|\tilde{A}) &\&=\langle n | \hat{\gamma}_+(\sum_{a=1}^n [x_a]) \hat{g}(\tilde{A}) | \lambda; n \rangle \cr
 &\& =   \frac{\langle 0 | \psi^\dag(x_n^{-1}) \cdots \psi^\dag(x_1^{-1}) \psi_{\lambda_1-1+n}(\tilde{A})
 \cdots \psi_{\lambda_n}(\tilde{A}) |0 \rangle}{\Big(\prod_{j=1}^n x_j\Big)\Delta(\xb)}  \cr
&\& = \frac{\det \Big(\langle 0 |  \psi^\dag(x_j^{-1}) \psi_{\lambda_k-k +n}(\tilde{A}) |0\rangle \Big)_{1\le j,k \le n}}{\left(\prod_{j=1}^n x_j\right)\Delta(\xb)}  \cr
&\& = \frac{\det\left(p_{\lambda_k - k + n} (x_j|\tilde{A})\right)_{1\le\, j,k \, \le n}}{\Delta(\xb)} = \tilde{s}_{\lambda, n}([\xb] | \tilde{A}),
\eea
where the third line follows from Wick's theorem (Appendix \ref{wick_app}, eq.~(\ref{wick_det})) and the fourth from (\ref{A_tilde_polynoms}).
\end{proof}

Besides the fact that the  $s_{\lambda, n}([\xb])|\tilde{A})$'s are polynomial KP $\tau$-functions expressible 
via the bialternant formula (\ref{s_lambda_tilde_A}),  they also share with the Schur functions $s_{\lambda}([\xb]))$
the property that, for $\lambda$ with Frobenius indices $(\alpha | \beta)$ as in (\ref{alpha_beta}),
 they satisfy the Giambelli  identity \cite{Mac1},  expressing  them as determinants of the matrices
 whose elements are the functions $s_{(\alpha_i|\beta_j)}([\xb]))$ corresponding to hook partitions 
  for all pairs $(\alpha_i, \beta_j)$.
 \begin{proposition}[Giambelli identity]
 \be
s_{\lambda, n}([\xb])|\tilde{A}) = \det \left(s_{(\alpha_i| \beta_j), n}([\xb])|\tilde{A}) \right)_{1\le i,j \le r}
\label{s_A_giambelli}
\ee
\end{proposition}
\begin{proof} 
From (\ref{pi_lambda_A_def}), (\ref{s_lambda_n_A_def}) and \ref{psi_alpha_beta_vac_n}), we have
\bea
s_{\lambda, n}([\xb])|\tilde{A})  &\&=(-1)^{\sum_{j=1}^r \beta_j} \langle n | \hat{\gamma}_+([\xb]) 
\prod_{j=1}^r \psi_{\alpha+j+n}(\tilde{A}) \psi^\dag_{-\beta_j +n -1}(\tilde{A}) |n \rangle \cr
&\&= (-1)^{\sum_{j=1}^r \beta_j} \langle n | 
\prod_{j=1}^r \hat{\gamma}_+([\xb]) \psi_{\alpha+j+n}(\tilde{A})\hat{\gamma}^{-1}_+([\xb]) 
 \hat{\gamma}_+([\xb]) \psi^\dag_{-\beta_j +n -1}(\tilde{A}) \hat{\gamma}^{-1}_+([\xb])  |n \rangle \cr
 &\&=  \det\left((-1)^{\beta_j} \langle n | xt
 \hat{\gamma}_+([\xb]) \psi_{\alpha+i+n}(\tilde{A})\psi^\dag_{-\beta_j +n -1}(\tilde{A}) |n \rangle \right)_{1\le i,j \le r} \cr
 &\&=  \det \left(s_{(\alpha_i| \beta_j), n}([\xb])|\tilde{A}) \right)_{1\le i,j \le r}
\eea
where the fact that $\hat{\gamma}^{-1}_+([\xb])$ stabilizes the vacuum has been used
in the second line and Wick's theorem (see Appendix \ref{wick_app})  in the third.
\end{proof}

\section{Generalized Nimmo formula for polynomial BKP $\tau$-functions}
\label{BKP_gen_nimmo}

Nimmo's formula \cite{Nim}  similarly expresses Schur $Q$-functions $Q_{\alpha}(\xb)$ associated to 
a strict partition $\alpha = (\alpha_1,\dots, \alpha_{2m})$ of even cardinality (possibly including 
a vanishing part $\alpha_{2m} =0$), as the ratio of two Pfaffians: 
\be
Q_{\alpha}(\xb) = 2^{2m}\frac{\Pf(M_{\alpha}(\xb))}{\Pf(M(\xb))},
\label{nimmo_Q_alpha}
\ee
where $\xb = (x_1, \dots, x_{2n})$ consists of an even number $2n$ of elements,
 $M(\xb)$ is the $2n \times 2n$ skew symmetric matrix
\be
M_{ab}(\xb) :=\frac{x_a-x_b}{x_a+x_b} \qquad 1\leq a, b \leq 2n
\label{M_ab_def}
\ee
and $M_{\alpha}(\xb)$ is the $2(n+m) \times 2(n+m)$ block skew symmetric matrix
\be
M_{\alpha}(\xb) := \begin{pmatrix}M(\xb) & V_{\alpha}(\xb)\\-V_{\alpha}(\xb)^{T} & 0 \end{pmatrix}, 
\ee
with
\be
\left(V_{\alpha}(\xb)\right)_{aj} := (x_a)^{\alpha_j}, \quad  a=1,\dots, 2n, \ j=1,\dots,2m.
\label{V_alpha_aj_def}
\ee
If the  number of elements is odd $(x_1, \dots, x_{2n-1})$, we just set $x_{2n}=0$ in \ref{nimmo_Q_alpha})- \ref{V_alpha_aj_def}).

Now assume that the infinite matrix $A$ appearing in eq.~(\ref{hat_A_def}) satisfies the antidiagonal
triangular conditions (\ref{A_antidiag_triang}) or, equivalently, that $\check{A}$
satisfies the strict upper triangular conditions (\ref{A_check_triang}),
so that  $\hat{A}^\pm$ annihilates the vacuum (\ref{A_triang_vac_annihil})
and $\hat{h}^\pm(A)$ stabilizes it (\ref{h_vac_stabil}).
It follows that the BKP $\tau$-function 
$Q_\alpha([\xb]_B| A)$ is a polynomial of degree $\le 2m$ (not necessarily homogeneous)
that  is expressible via a generalized Nimmo formula.
Define the $2m \times 2m$ skew symmetric matrix matrix $H_\alpha(A)$ with elements
\bea
H_\alpha(A)_{jk} &\&:=
\begin{cases} 
\langle 0| \phi^\pm_{\alpha_j}(A) \phi^\pm_{\alpha_k}(A) |0 \rangle, \quad 1\le j<k \le 2m,  \cr
0 \ \text{ if } j=k.
\end{cases} \cr
&\& \cr
&\& = - H_\alpha(A)_{kj} \\
&\& = \sum_{i=1}^{\alpha_k}(-1)^i P_{-i, \alpha_j}(A)P_{i, \alpha_k}(A) +2 P_{0, \alpha_j}P_{0,\alpha_k},  \quad 1\le j<k \le 2m.
\label{H_alpha_A_def}
\eea
\begin{proposition}[Generalized Nimmo formula]
\label{general_nimmo_prop}
Assuming condition (\ref{A_antidiag_triang}) to hold, we have
\be
Q_{\alpha}(\xb| A) = 2^{2m}\frac{\Pf(M^H_{\alpha}(\xb|A))}{\Pf(M(\xb))},
\label{gen_nimmo_A}
\ee
where
\be
M^H_{\alpha}(\xb | A) = \begin{pmatrix}M(\xb) & V_{\alpha}(\xb| A)\\-V_{\alpha}(\xb|A)^{T} & 2 H_\alpha(A) \end{pmatrix}, 
\label{M_alpha_def}
\ee
with 
\be
\left(V_\alpha(\xb|A)\right)_{aj}= p_{\alpha_j}(x_a|A), \quad  a=1,\dots, 2n, \quad j=1,\dots, 2m.
\label{V_alpha_def}
\ee
\end{proposition}
\begin{remark}
Note that the numerator Pfaffian $\Pf(M^H_{\alpha}(\xb|A))$ in (\ref{gen_nimmo_A}) vanishes whenever 
any pair $x_a = x_b$ are equal, and hence  we may factor out a Vandermonde determinant $\Delta(x_1, \dots , x_{2n})$.
The denominator Pfaffian is
\be
\Pf(M(\xb)) = \prod_{1 \le a < b \le 2n} \frac{x_a-x_b}{x_a+x_b},
\label{pfaff_M_eval}
\ee
so the $\Delta(x_1, \dots , x_{2n})$  factors in the numerator and denominator cancel. This also places a
factor $ \prod_{1 \le a < b \le 2n} (x_a+x_b)$ in the numerator, which cancels the poles  from the matrix  elements $M_{ab}(\xb)$ at which $x_a + x_b$ vanishes  for any distinct pair $(a,b)$.
Therefore there are no poles, and the result is a polynomial which, since both the numerator and denominator
reverse signs under any interchange $x_a \leftrightarrow x_b$, is symmetric.
\end{remark}
\begin{proof}{(Proposition \ref{general_nimmo_prop})}
We have the standard formula \cite{DJKM1}, 
\be
\langle 0|  \phi^\pm(-x_{2n}^{-1}) \cdots \phi^\pm(-x_{1}^{-1})  = 2^{-n} \Pf(M(\xb)) \langle 0|\hat{\gamma}^{B\pm}( 2[\xb]_B) 
\label{phi_bosonization}
\ee
related to the bosonization map. From the fermionic VEV formula (\ref{A_polynoms} for the polynomials $ p_j(x_a|A)$,
 we have
\be
\left(V_{\alpha}(\xb | A)\right)_{aj} := \langle 0 | \phi^\pm(-x_a^{-1})\phi^\pm_j(A)| 0 \rangle 
 = p_j(x_a|A), \quad 1 \le a \le 2n, \ j\in \Nb,
\label{p_j_A_VEV}
\ee
and from (\ref{phi_pairing})
\be
\langle 0 |\phi^\pm(-x_a^{-1})\phi^\pm (-x_b^{-1})|0 \rangle =\tfrac 12 \tfrac {x_a-x_b}{x_a+x_b},\quad 1 \le a,b \le 2n.
\label{M_met_els}
\ee
Substituting (\ref{phi_bosonization}) and (\ref{phi_pm_A_def})  into (\ref{kappa_alpha_A_def}),  (\ref{Q_A_QQ_A}) and (\ref{Q_alpha_A_def})
and using the fact  that $\hat{h}(A)$ stabilizes the vacuum (\ref{h_vac_stabil}) gives
\bea
Q_\alpha(\xb |A) &\&= 2^{m}\langle 0 | \hat{\gamma}^{B\pm} (2[\xb]_B) \hat{h}^\pm(A) |\alpha^\pm) 
 \\
&\& = 2^{m}\langle 0 | \hat{\gamma}^{B\pm}(2[\xb]_B) \phi^\pm_{\alpha_1}(A) \cdots \phi^\pm_{\alpha_{2m}}(A)  |0\rangle
\label{Q_alpha_VEV_prod_phi_A}
 \\
&\& \cr
&\& = 2^{n+m}\frac{ \langle 0|  \phi(-x_{2n}^{-1}) \cdots \phi(-x_{1}^{-1}) \phi^\pm_{\alpha_1}(A) \cdots \phi^\pm_{\alpha_{2m}}(A)  |0\rangle}  {\Pf(M(\xb))}\\
&\& \cr
&\& =2^{n+m} \frac{\Pf\begin{pmatrix}\tfrac{1}{2}M(\xb) & V_{\alpha}(\xb| A)\\-V_{\alpha}(\xb|A)^{T} & H_\alpha(A)\end{pmatrix}} {\Pf(M(\xb))} \cr
&\& =2^{2m} \frac{\Pf\begin{pmatrix}M(\xb) & V_{\alpha}(\xb| A)\\-V_{\alpha}(\xb|A)^{T} & 2 H_\alpha(A)\end{pmatrix}} {\Pf(M(\xb))} , 
\label{Q_alpha_x_B_A_nimmo}
\eea
where Wick's theorem (Appendix \ref{wick_app}, eq.~(\ref{wick_pfaff})) has been applied in the fourth line, and the matrix elements evaluated  using  relations (\ref{H_alpha_A_def}), (\ref{p_j_A_VEV}) and (\ref{M_met_els}). 
\end{proof}

We also have an analog of Schur's Pfaffian formula
\be
Q_\alpha(\xb) = \Pf\left(Q_{ij}(\xb)\right)_{1\le i, j \le 2m},
\ee
for the functions $Q_\alpha(\xb | A)$, where $Q(\xb)$ is the skew-symmetric $2m \times 2m$ matrix with elements
\bea
Q_{ij}(\xb) &\&= \begin{cases}
Q_{(i,j)}([\xb]_B)   \qquad  \text{if } i > j \cr
\quad 0 \qquad \qquad  \quad   \text{ if } i =j \cr
-Q_{(j,i)}(\xb)  \, \qquad \text{ if } i < j, 
\end{cases}\\
&\&\cr
&\&\phantom{x} {\hskip16 pt}1\le i,j \le  2m.
 \nonumber
\eea

\begin{proposition}
\be
Q_\alpha(\xb |A) = \Pf\left(Q_{ij}(\xb|A)\right)_{1\le j \le 2m},
\ee
where $Q(\xb | A)$ is the skew-symmetric $2m \times 2m$ matrix with elements
\be
Q_{ij}(\xb |A) = \begin{cases}
Q_{(i|j)}(\xb|A)   \qquad  \text{if } i < j \cr
\quad 0 \qquad \qquad  \quad   \text{ if } i =j \cr
-Q_{(i|j)}(\xb|A)   \quad \text{ if } i > j 
\end{cases}, \quad 1\le i,j, \le  2m.
\ee
\end{proposition}
\begin{proof}
This follows from substituting eq.~(\ref{alpha_pm_state}  into eqs.~(\ref{kappa_alpha_A_def},
 (\ref{Q_alpha_A_def}) and (\ref{Q_A_QQ_A} and using (\ref{phi_pm_A_def}) to express
 $Q_\alpha(\xb|A)$, as in eq.~(\ref{Q_alpha_VEV_prod_phi_A})
 \bea
 Q_\alpha(\xb|A) &\&=  2^{m}\langle 0 | \hat{\gamma}^{B\pm}(2[\xb]_B) \phi^\pm_{\alpha_1}(A) \cdots \phi^\pm_{\alpha_r}(A) | 0 \rangle \cr
&\& = 2^{m}\langle 0 | \hat{\gamma}^{B\pm}(2[\xb]_B) \phi^\pm_{\alpha_1}(A) \big(\hat{\gamma}^{B\pm}\big)^{-1}(2[\xb]_B)
  \cdots  \hat{\gamma}^{B\pm}(2[\xb]_B) \phi^\pm_{\alpha_r}(A) \big(\hat{\gamma}^{B\pm}\big)^{-1}(2[\xb]_B)| 0 \rangle \cr
  &\& = \Pf\left( 2\langle 0 | \hat{\gamma}^{B\pm}(2[\xb]_B) \phi^\pm_{\alpha_i}(A) \phi^\pm_{\alpha_j}(A) |0 \rangle \right)_{1\le i, j \le 2m}\cr
&\&=  \Pf\left(Q_{ij}(\xb |A)\right)_{1\le i, j \le 2m},
 \eea
 where the fact that $\hat{\gamma}^{B\pm}([\xb]_B)$ stabilizes the vacuum has been used in
 the second line and Wick's theorem (Appendix \ref{wick_app}, eq.~(\ref{wick_pfaff}) in the third.
\end{proof}
If we further choose the matrix $A$ to satisfy the vanishing conditions 
\be
A_{jk}= 0 \ \text{ if } j<0, \ k\le 0
\label{null_neg_A}
\ee
or equivalently, $\check{A}$ to satisfy 
\be
\check{A}_{jk}= 0 \ \text{ if } j<0, \ k\ge 0,
\label{null_neg_check_A}
\ee
we have
\be
P_{jk}(A) =0 \ \text{ if } j<0, k\ge 0, 
\label{P_neg_null}
\ee
and it follows that, for $j\ge 0$,  $\phi_j(A)$ is the finite triangular linear combination
\be
\phi^\pm_j(A) = \sum_{k=0}^jP_{kj}(A) \phi^\pm_k
\label{phi_j_A_series}
\ee
of  $\phi_k$'s  with $k\ge 0$.
This implies that 
\be
\langle 0 | \phi^\pm_j(A) \phi^\pm_k(A) | 0 \rangle  =0 , \quad \forall \ j, k\ge 0,  \ (j,k)\neq (0,0),
\label{G__jk_null}
\ee
so the skew matrix $H_\alpha(A)$  defined in (\ref{H_alpha_A_def}) vanishes and 
eq.~(\ref{M_alpha_def}) becomes
\be
M^0_{\alpha}(\xb | A) = \begin{pmatrix}M(\xb) & V_{\alpha}(\xb| A)\\-V_{\alpha}(\xb|A)^{T} & 0 \end{pmatrix}.
\label{M_0_alpha_def}
\ee

A particular case of such generalized Schur $Q$-functions, introduced in \cite{Iv}, corresponds 
to choosing the monic polynomials $\{p_j(x|A)\}_{j\in \Nb}$ as
\be
p_j(x|A_{int}(\ab)) := (x|\ab)^{(j)} = \prod_{i=1}^j (x-a_i), \quad j\in \Nb^+, \quad p_0(x, A_{int}(\ab)) := (x|\ab)^{(0)} = 1.
\ee
where $\ab :=(a_1, a_2, \dots)$  is an arbitrary sequence  of shared roots.
These were referred to as {\em interpolation}  analogs of the  Schur $Q$-functions in \cite{Iv} and shown 
to be BKP $\tau$-functions in \cite{Roz}.

\section{KP multipair and  BKP multipoint correlators}
\label{KP_BKP_correls}


\subsection{KP $n$-pair correlators}
\label{KP_correls}

When $\tb$ is restricted to equal the difference between two finite (normalized) power sums
\be
t_j =\tfrac{1}{j}\sum_{a=1}^{n} (x_a^j - y_a^j), \quad j \in \Nb^+
\ee
in terms of two sets of $n$ variables
\be
\xb :=(x_1, \dots, x_n), \quad \yb :=(y_1, \dots, y_n),
\ee
we denote  it as 
\be
{\bf t}=[\xb,\yb] := [\xb] - [\yb].
\label{tb_xb_yb}
\ee
For a partition $\lambda$ and group element $g(\tilde{A})$,  define the $n$-pair correlation  function
\be
K_{n,\lambda} (\xb,\yb|\tilde{A}) :=  \langle 0|\Big(\prod_{a=1}^n\psi^\dag(x_a^{-1})\psi(y_a^{-1})\Big)\hat{g}(\tilde{A}) | \lambda \rangle.
 \label{n_pair_correl_lattice} 
\ee
From the identity ({\hskip -2.5 pt} \cite{HB}, Chapt.~5.9)
\be
\langle 0|\Big(\prod_{a=1}^n\psi^\dag(x_a^{-1})\psi(y_a^{-1})\Big) 
=\big(\prod_{a=1}^n x_a y_a\big) \Delta(\xb, \yb) \langle 0|  \hat{\gamma}_+([\xb] - [\yb]), 
\ee
where
\be
\Delta(\xb,\yb):= \det\Big(\frac{1}{x_a-y_b}\Big)_{1\le a,b, \le n}
=\frac{(-1)^{\tfrac{1}{2}n(n-1)}\Delta(\xb) \Delta(\yb)}{\prod_{1\le a,b\le n}(x_a -y_b)}.
\label{Delta_def}
\ee
we have
\bea
K_{n,\lambda} (\xb,\yb |\tilde{A}) 
= \big(\prod_{a=1}^n x_a y_a\big)\Delta(\xb,\yb) \pi_\lambda(g(\tilde{A}))([\xb,\yb]).
 \label{KP_tau_lattice_xy}
\eea

Now choose $\tilde{A}$ to be strictly upper triangular, so $\hat{g}(\tilde{A})$
stabilizes the vacuum, as in (\ref{g_vac_stabil}). 
Define the $(n+r )\times (n+r)$  matrix 
 \be
\Nb^G_{n,\lambda}(\xb, \yb|\tilde{A}) :=\begin{pmatrix}N_n(\xb, \yb) & W_{n,\alpha}(\xb| \tilde{A})\\
W^*_{n,\beta}(\yb|\tilde{A})^{T} &  G_{(\alpha|\beta)}(\tilde{A})) \end{pmatrix},
\label{N_n_lambda_ratl}
 \ee
consisting of four blocks: a left upper  $n \times n$ block
 $N_n(\xb, \yb |\tilde{A})$ that is independent of $\tilde{A}$ and $\lambda$:
 \bea
\big(\Nb^G_{n,(\alpha|\beta)}(\xb, \yb|\tilde{A}\big)_{ab} &\&= \langle 0|\psi^\dag(x_a^{-1})\psi(y_b^{-1}) |0\rangle =:\frac{x_ay_b}{x_a-y_b}, \cr
 &\& \cr
&\&:= (N_n(\xb, \yb))_{ab},  \quad 1\le a, b  \le n.
 \eea
 a  right lower  $r\times r$ block
\bea
 \big(G_{(\alpha|\beta)}(\tilde{A})\big)_{ij} &\&
 =(-1)^{\beta_i} \langle 0|\psi_{\alpha_j}(\tilde{A}) \psi^*_{-\beta_i-1}(\tilde{A})|0\rangle, \quad 1\le i,j \le r \cr
 &\& = (-1)^{\beta_i} \sum_{k=0}^{\beta_i} P_{{-k-1}, \alpha_j}(\tilde{A}) P^*_{\beta_i,  k} (\tilde{A}),
 \label{G_lambda_def} 
 \eea
that is independent of $(\xb, \yb)$, and  two rectangular  $n\times r$ and $r\times n$  blocks
given by the polynomials defined in eqs.~(\ref{A_tilde_polynoms}), (\ref{A_tilde_polynoms_dual})
\bea
\big(W_{n, \alpha}(\xb | \tilde{A})\big)_{aj} &\&:= 
\langle 0 | \psi^\dag(x_a^{-1}) \psi_{\alpha_j}(\tilde{A})|0\rangle = x_a p_{\alpha_j}(x_a |\tilde{A}), \\
\big(W^*_{n, \beta}(\yb | \tilde{A})\big)_{aj} &\&=  \langle 0 | \psi(y_a^{-1}) \psi^*_{-\beta_j -1}(\tilde{A})|0\rangle 
 = y_a p^*_{\beta_j}(y_a |\tilde{A}),
\label{W_n_alpha_beta_def}
\\
&\&1\le a \le n, \ 
1\le j \le r.
\nonumber
\eea
Wick's theorem  then again implies a finite determinantal representation of 
$\pi_\lambda(g(\tilde{A}))([\xb,\yb])$ and hence of $K_{n,\lambda} (\xb,\yb |\tilde{A})$.
\begin{proposition}
\label{K_n_pair_correl}
The $n$-pair correlators are  given by the determinantal formula 
\be
K_{n,\lambda} (\xb,\yb |\tilde{A})= \det \left(\Nb^G_{n,\lambda}(\xb, \yb|\tilde{A})\right),
\label{n_pair_correl_det_A_tilde}
\ee
and therefore are rational functions,  with simple poles  at the points where $x_a - y_b$ vanish,
  \end{proposition}
  \begin{remark}
  Whenever two $x_a$'s or two $y_a$'s are equal,  two of the rows or columns of (\ref{N_n_lambda_ratl_restr})
 coincide, and the determinant (\ref{n_pair_correl_det_A_tilde}) vanishes. Therefore, we can factor 
 the product $\Delta(\xb)\Delta(\yb)$ of
 Vandermonde determinants from  the expression for $K_{n,\lambda} (\xb,\yb |\tilde{A})$. It follows from (\ref{KP_tau_lattice_xy})
 that 
 \be
 \big(\prod_{a=1}^nx_a y_a\big) \pi_\lambda(g(\tilde{A}))([\xb,\yb]) =\frac{ K_{n,\lambda} (\xb,\yb |\tilde{A})}{\Delta(\xb, \yb)}
 \ee
  is a polynomial in the position variables $\{x_a, y_a\}_{1\le a \le n}$ of degree no greater than $2n +|\lambda| -r$.
  \end{remark}
  \begin{proof}{(Proposition \ref{K_n_pair_correl})} 
  Eqs.~(\ref{n_pair_correl_lattice}) and (\ref{psi_alpha_beta_vac_n}) imply
   \bea
  K_{n,\lambda} (\xb,\yb|\tilde{A}) &\&=  (-1)^{\sum_{j=1}^r \beta_j }  \langle 0|\prod_{a=1}^n\psi^\dag(x_a^{-1})\psi(y_a^{-1})
 \prod_{j=1}^r \psi_{\alpha_j}(\tilde{A}) \psi^\dag_{-\beta_j -1} (\tilde{A})|0 \rangle \cr
 &\& = \det \left(\Nb^G_{n,\lambda}(\xb, \yb|\tilde{A})\right),
  \eea
  where Wick's theorem (Appendix \ref{wick_app}, eq.~(\ref{wick_det})) and   Lemma \ref{VEV_p_j} have been applied in the second line.
  \end{proof}
 If we furthermore choose $\tilde{A}$ such that
   \be
   \tilde{A}_{jk} = 0 \ \text { if } j<0, \ k \ge 0, \ra  P_{jk}(\tilde{A}) = 0 \ \text{ if } j<0, \ k\ge 0,
   \label{A_tilde_null_cond}
   \ee
   which can always be done without changing the polynomials $\{p_j(x |\tilde{A})\}_{j\in \Nb}$ defined in (\ref{p_j_tilde_A_def}),
it follows from eq.~(\ref{G_lambda_def}),  that 
\be
 \big(G_{\lambda}(\tilde{A})\big)_{ij}  =0, \ \forall \ i,j \in \Nb.
 \ee
 and hence
 \be
\Nb^0_{n,\lambda}(\xb, \yb|\tilde{A}) =\begin{pmatrix}N_n(\xb, \yb) & W_{n,\alpha}(\xb| \tilde{A})\\W_{n,\beta}(\yb|\tilde{A})^{T} &  0 \end{pmatrix}.
\label{N_n_lambda_ratl_restr}
 \ee


\subsection{BKP $2n$-point correlators}
\label{BKP_correls}

 Let $\xb =(x_1, \dots, x_{2n})$ with all $x_a$'s distinct, and define the $2n$-point BKP correlation function
 \be
 K^B_{2n,\alpha}(\xb):= \langle 0 | \phi^\pm(-x_{2n}^{-1}) \cdots  \phi^\pm(-x_{1}^{-1})\hat{h}^\pm(A) | \alpha^\pm),
 \label{BKP_2n_point_def}
 \ee
 Eqs.~(\ref{phi_bosonization})  and (\ref{Q_alpha_A_def}) imply that $K^B_{2n,\alpha}(\xb)$ is related to
 $Q_\alpha([\xb]_B|A)$ by
 \be
K^B_{2n,\alpha}(\xb)  = 2^{-2n-2m} \Pf(M(\xb))Q_\alpha(\xb|A)= 2^{-2n-2m}\Pf(M^G_{\alpha}(\xb|A)),
 \ee
 where $Q_\alpha([\xb]_B|A)$ is given by the generalized Nimmo formula (\ref{gen_nimmo_A})
 and $M^H_{\alpha}(\xb|A)$ is defined by (\ref{M_alpha_def}) or, if the additional 
 conditions (\ref{null_neg_A}) - (\ref{P_neg_null}) are fulfilled, by (\ref{M_0_alpha_def}).


\section{Examples}
\label{examples}
A large number of examples of lattices of polynomial KP Schur functions, labelled by partitions $\lambda$,
that can be expressed in terms of systems of monic polynomials 
 $\{p_j(x |\tilde{A})\}_{j\in \Nb}$ via formula (\ref{s_lambda_tilde_A}) are detailed in \cite{HL, SV, Mac2}.   
 These include:  several of the various types of generalized Schur functions studied in  \cite{Mac2}, 
  the {\em factorial}\,  Schur functions  \cite{BL1, BL2},  and orthogonal and symplectic  group characters \cite{HL, SV}. 

Two cases of particular interest are the generalized Laguerre multivariate polynomials \cite{La, Ol},  
given by the bialternant formula (\ref{s_lambda_tilde_A}), with the 
system of monic polynomials $\{p_j(x | \tilde{A})\}_{j\in \Nb^+}$ chosen to be the associated Laguerre polynomials, 
and a $2$-parameter generalization of these \cite{HO2, Ol}, which we refer to as {\em double Laguerre polynomials}.
\begin{example}

Consider  the special choice (see \cite{HO2, OrSc})
\bea
 \tilde{A}^\rb(\pb)_{jk}&\&= p_{k-j} r_{j+1} \cdots r_k, \quad k>j 
 \label{A_r_p_jk},
 \eea
  parametrized  by two sets of  infinite numbers of parameters ${\bf r}:=\{r_j\}_{j \in \Zb}$ and  ${\bf p}=(p_1,p_2,p_3,\dots)$.
This gives rise to the following polynomial expression for the KP $\tau$-function defined
 in eqs.~(\ref{pi_lambda_A_def}), (\ref{s_lambda_A_def})  (see Example 4.4 in \cite{HO2}),
\be
s_\lambda(\tb | \tilde{A}^\rb(\pb))  = s_\lambda(\tb)+
\,\sum_{\rho < \lambda}
\,r_{\lambda/\rho}\,s_{\lambda/\rho}({\bf p})\,s_{\rho}(\tb),
\label{s_lambda_r_p_expansion}
 \ee
 where $s_{\lambda/\rho}({\bf p})$  is the skew Schur function corresponding to the skew
partition $\lambda/\rho$, and
\be
 r_{\lambda/\rho} :=\prod_{i=1}^{\ell(\lambda)} \prod_{j=\rho_i+1}^{\lambda_i}r_{j-i} =\frac{ r_\lambda}{ r_\rho}
 \label{skew_content_prof_r}
\ee

Evaluating at  $\tb = [\xb]$, the $\tau$-function $s_\lambda([\xb] | \tilde{A}^\rb(\pb))$ is given by the bialternant 
formula (\ref{s_lambda_tilde_A}), in  which the  system of monic polynomials  
(\ref{p_j_A_tilde_x}), (\ref{p_star_j_A_tilde_x}) is:
\bea
p_j(x | \tilde{A}^{\rb}(\pb))&\& := x^{-1}\langle 0|\psi^\dag(x^{-1})\hat{g}(\tilde{A}^\rb(\pb)) \psi_j|0\rangle =
x^j +\sum_{k=1}^{j}   r_j \cdots r_{j-k+1} s_{(k)}({\bf p}) x^{j-k}
\label{p_j_p_r} 
\\
p_j^*(y| \tilde{A}^{\rb}(\pb))&\& := y^{-1} \langle 0|\psi(y^{-1})\hat{g}(\tilde{A}^\rb(\pb)) \psi^\dag_{-j-1}|0\rangle =
\ y^j+\sum_{k=1}^j   r_{-j} \cdots r_{k-j-1} s_{(k)}(-{\bf p}) y^{j-k} .
\label{p_star_j_p_r} \cr&\&
\eea 

  We also have
\bea
\langle 0|\hat{g}(\tilde{A}^\rb(\pb)) \psi_i\psi^\dag_{-j-1}|0\rangle
&\& =
 r_{-j}r_{1-j}\cdots r_{i}  \langle 0|\hat{\gamma}_+(\pb)\psi_i\psi^\dag_{-j-1}|0\rangle \cr
 &\& =(-1)^j r_{-j}r_{1-j}\cdots r_{i} s_{(i|j)}(\pb).
 \label{G_ij_A}
 \eea
Therefore, choosing $\tb=[\xb] - [\yb]$, as in (\ref{tb_xb_yb}), the resulting $n$-pair correlation function
$K_{n,\lambda} (\xb,\yb |\tilde{A}^\rb(\pb))$  is given by (\ref{n_pair_correl_det_A_tilde}),
where the system of polynomials defining the rectangular matrices
$W_{n, \alpha}(\xb | \tilde{A}^\rb(\pb)$ and  $W^*_{n, \beta}(\yb | \tilde{A}^\rb(\pb))$ appearing in
the $(n+m) \times (n+m)$ matrix $\Nb^G_{n,\lambda}(\xb, \yb|\tilde{A}^\rb(\pb))$
defined by (\ref{N_n_lambda_ratl}) is  (\ref{p_j_p_r}), (\ref{p_star_j_p_r}) 
and the matrix elements of $G_{\alpha,\beta}(\tilde{A}^r({\bf p}))$ are
\be
\left(G_{\alpha,\beta}(\tilde{A}^r({\bf p}))\right)_{ij}:= 
(-1)^{\beta_i}\langle 0|\hat{g}(\tilde{A}^\rb(\pb)) \psi_{\alpha_j}\psi^\dag_{-\beta_i-1}|0\rangle
=r_{-\beta_i}\cdots r_{\alpha_j} s_{(\alpha_j|\beta_i)}({\bf p}) . 
\label{G_alpha_beta_A}
\ee

\begin{remark}
The multivariate  Laguerre symmetric polynomials \cite{Ol, La}  are  the
 special case of (\ref{s_lambda_r_p_expansion}) obtained  by  choosing
\bea
\pb &\& =\tb_0:=(1,0,0,\dots)
\label{p_0_def}
\\
r_j&\&=r_j(z):=-j (z+j).
\label{r_j_laguerre}
\eea
 Then
(\ref{p_j_p_r}), (\ref{p_star_j_p_r}) reduce to the usual  associated Laguerre
polynomials (normalized to be monic)
\bea
 p_j(x | \tilde{A}^{\rb(z)}(\tb_0)) &\&=  
 \sum_{k=0}^{j}(-1)^{j-k}  {z+j \choose z+k}{j \choose k} x^{k} \cr
&\&= (-1)^j j!L_j^{(z)}(x),
\label{p_j_Lag}
\cr
&\& \\
 p_j^*(y| \tilde{A}^{\rb(z)}(\tb_0)) &\&=   
 \sum_{k=0}^{j}(-1)^{j-k}  {-z+j \choose -z+k}{j \choose k} y^{k}\cr
&\&=  p_j^*(y| \tilde{A}^{\rb(-z)}(\tb_0)).
\label{p_j_Lag_star} \cr
&\&
\eea
For this case, since $r_0=0$, it follows from (\ref{G_alpha_beta_A}) that the matrix $G_{(\alpha|\beta)}(\tilde{A}^{\rb(z)}(\pb))$
vanishes.

The multivariate  {\em double} Laguerre symmetric polynomial  \cite{Ol} (see also Remark 4.2 in \cite{HO2})  is 
obtained from (\ref{s_lambda_r_p_expansion}) by  choosing
\bea
\pb &\& =\tb_0:=(1,0,0,\dots)
\label{p_0_def_2}\\
r_j&\&=r_j(z,z'):=-(z+j)(z'+j),
\label{r_j_laguerre_2}
\eea
so that
\be
r_j(z) = r_j(z,0).
\ee
Then
(\ref{p_j_p_r}), (\ref{p_star_j_p_r}) reduce to 
\bea
p_j(x | \tilde{A}^{\rb(z, z')}(\tb_0)) &\&=  
\sum_{k=0}^{j}(-1)^{j-k}  {z+j \choose z+k}{z'+j \choose z'+k} x^{k},
\label{p_j_double_Lag}
\cr
&\& \\
p_j^*(y| \tilde{A}^{\rb(z, z')}(\tb_0)) &\&=   
\sum_{k=0}^{j}(-1)^{j-k}  {-z+j \choose -z+k}{-z'+j \choose -z'+k} y^{k}\cr
&\&= p_j(x | \tilde{A}^{\rb(-z, -z')}(\tb_0)).
\label{p_j_double_Lag_star} \cr
&\&
\eea

\end{remark}

We also have the generalized Schur $Q$-functions  $Q_\alpha(\xb | A^{\rb}(\pb_B))$,
with (see \cite{HO2, Or})
\bea
A^{\rb}(\pb_B)_{jk} &\&=\begin{cases}  (-1)^k \, p_{-j-k}\, r_{j+1} \cdots r_{-k} , \quad j+k < 0 
  \text{ and odd}\cr
= 0 \quad \text{ if } j+k \ \text{ is even or } j+k \ge 0,
  \end{cases}
  \cr
  &\&
 \label{A_r_p_B_jk}
\eea
given by the generalized Nimmo formula (\ref{gen_nimmo_A}),  with 
 the rectangular matrix $V_\alpha(\xb |A^{\rb}(\pb))$ defined in (\ref{V_alpha_def}) given by 
the polynomials 
\bea
p_j(x|(A^\rb(\pb_B))) &\&:=  \langle 0|\phi^\pm(-x^{-1})\phi^\pm_j(A^\rb(\pb_B)) |0\rangle   \cr
&\&=x^j+\sum_{k=1}^i r_{j}  \cdots r_{j-k+1}  x^{j-k} q_k(\tfrac 12 \pb_B).
\label{B-pol}
\eea
which coincide with  the $p_j(x | \tilde{A}^{\rb}(\pb))$'s of eq.~(\ref{p_star_j_p_r}) when $\pb=(p_1, 0, p_3, 0 , \dots , 0 ,p_{2j-1}, \dots)$,
and the matrix elements of $H_{\alpha}({A}^r({\bf p}_B))$  are
\bea
\left(H_{\alpha}({A}^\rb({\bf p}_B))\right)_{ij} &\&:=  
\langle 0| \phi_{\alpha_i}(\tilde{A}^\rb(\pb_B)) \phi_{\alpha_j}(\tilde{A}^\rb(\pb_B))  |0\rangle  \cr
&\&=\frac 12 r_1 \cdots r_{\alpha_i} r_1\cdots r_{\alpha_j} Q_{(\alpha_i,\alpha_j)}( \tfrac 12 {\bf p}_B)  . 
\eea

Similarly to (\ref{s_lambda_r_p_expansion}), this can be expressed as a sum over ordinary (scaled) 
Schur $Q$-functions $\QQ_\beta$,
with $\beta\leq \alpha$
\be
\QQ_{\alpha}([\xb]_B | {A}^{\rb}({\bf p}_B)) = \QQ_{\alpha}({[\xb]_B})+
\sum_{\beta < \alpha}\, 
r^B_{\alpha /\beta} \,\QQ_{\alpha/\beta}({\bf p}_B)\,\QQ_{\beta}({[\xb]_B}),
\label{Q_alpha_r_p_expansion}
\ee
 where 
\be
\QQ_{\alpha/\beta}([\pb]_B) = (\alpha|\gamma^{B}([\pb]_B) | \hat{\beta})
\ee
  is the (scaled) skew Schur Q-function corresponding to the skew partition $\alpha/\beta$  and
\be 
 r^B_{\alpha/\beta}   :=\prod_{i=1}^{2m} \prod_{j=\beta_i+1}^{\mu_i} r_j =\frac{r^B_{\alpha}}{ r^B_\beta}
 \label{skew_strict_content_prod_r}
\ee 
is the content product over nodes of the {\em shifted} skew Young diagram  \cite{Iv} of $\alpha/\beta$.
 
\end{example}

\begin{example}
For KP single pair correlation functions with $\Nb^0_{n,\lambda}(x, y|\tilde{A}_{fac})$ of
the form (\ref{N_n_lambda_ratl_restr}), for $n=1$, if the Frobenius rank $r$ of the partition function $\lambda$ 
is greater than $1$ the determinant in (\ref{n_pair_correl_det_A_tilde}) vanishes, so the only nonvanishing pair
correlators are the trivial one, which is
\be
K_{2, \emptyset} (x,y) | \tilde{A}) =\frac{xy}{y - x}
\ee
and the one for a hook partition, which is
\be
K_{2, (\alpha_1|\beta_1))} (x,y) | \tilde{A}) = -p_{\alpha_1}(x|\tilde{A}) p_{\beta_1}(y | \tilde{A}).
\ee
If we do not impose the null conditions (\ref{A_tilde_null_cond}) on $\tilde{A}$, so  the matrix
$G_{(\alpha|\beta)}(\tilde{A})$ is, in general, nonvanishing, we have
\be
K_{2, (\alpha_1|\beta_1))} (x,y) | \tilde{A}) 
= \gamma\frac{xy}{y - x} -p_{\alpha_1}(x|\tilde{A}) p_{\beta_1}(y | \tilde{A}).
\ee
where
\be
\gamma = \big(G_{(\alpha_1, \beta_1)}(\tilde{A})\big)_{11}
\ee
\end{example}

\begin{example}
For BKP $2$-point correlation functions $K^B_{2, \alpha}(x_a,x_2)|A)$  with $\Mb^0_\alpha(x_1,x_2)|A)$ of the form (\ref{M_0_alpha_def}),  if the cardinality $r=2m$ is greater than $2$, $\Mb^0_\alpha(x_1,x_2)|A)$,   which can have rank at most $4$, 
has vanishing Pfaffian. Therefore, the only nonvanishing $2$-point correlations functions are the one for the trivial partition
\be
K^B_{2, \emptyset} (x_1,x_2) | \tilde{A}) = \frac{x_1-x_2}{x_1+x_2}
\ee
and  the one where $\alpha$ has two parts
\be
K^B_{2, (\alpha_1, \alpha_2)} (x_1, x_2) | A) = p_{\alpha_1}(x_2)p_{\alpha_2}(x_1) -  p_{\alpha_1}(x_1)p_{\alpha_2}(x_2). 
\ee
If we do not impose the null conditions (\ref{null_neg_A})- (\ref{P_neg_null}) on $A$, so the matrix
$H_{\alpha}(\tilde{A})$ is, in general, nonvanishing, we have
\be
K^B_{2, (\alpha_1, \alpha_2)} (x_1,x_2) | \tilde{A}) =\sigma\frac{x_1-x_2}{x_1+x_2}
+ p_{\alpha_1}(x_2)p_{\alpha_2}(x_1) -  p_{\alpha_1}(x_1)p_{\alpha_2}(x_2),
\ee
where
\be
\sigma=  \big(H_{(\alpha_1, \alpha_2)}(A)\big)_{12}.
\ee
\end{example}

\begin{appendix}
\section{Appendix. Wick's theorem}
\label{wick_app}

The following summarizes the implication of Wick's theorem for fermionic VEV's
in a form  that is used repeatedly in this work ({\hskip - 2.5 pt} see,  e.g., \cite{HB}, Chapt. 5.1).

 \begin{theorem}[Wick's theorem]
 \label{wick}
The  vacuum expectation value of the product  of an even number of  linear elements $\{w_i\}_{1\le i \le 2n}$
 of the  fermionic Clifford algebra is
\be
\langle 0| w_1w_2\cdots w_{2n}|0\rangle= 
\Pf\left(W_{ij}\right)_{1\le i,j \le 2n}
\label{wick_pfaff}
\ee
where $\{W_{ij}\}_{1\le i,j \le 2n}\}$ are  elements of the skew $2n \times 2n$ matrix defined by
\be
W_{ij} = \begin{cases}
\langle 0| w_i w_j|0\rangle \quad \ \ \text{ if } i<j \cr
\qquad 0  \quad \qquad  \  \ \text{ if } i=j \cr
-\langle 0| w_i w_j|0\rangle \quad \text{ if } i> j
\end{cases}, 
\ee
whereas the VEV of the product of an odd number vanishes
\be
\langle 0| w_1w_2\cdots w_{2n+1}|0\rangle = 0.
\ee
 \end{theorem}
 In particular, if  half the $w_i$'s consist of  creation operators $\{u_i\}_{i=1, \dots , n}$
 and the other  half annihilation operators  $\{v^\dag_i\}_{i=1, \dots , n}$, so that
 \be
 \langle 0| u_i u_j |0 \rangle =0, \quad   \langle 0| v^\dag_i v^\dag_j |0 \rangle =0,\quad 1\le i,j \le n,
 \ee
 then (\ref{wick_pfaff})  reduces to
 \be
 \langle 0| u_1 v^\dag_1 \cdots u_n v^\dag_n|0\rangle=  \det \left(\langle 0 | u_i v^\dag_j |0 \rangle\right)_{1\le i, j \le n}.
 \label{wick_det}
 \ee

\end{appendix}

\bigskip
\bigskip
\noindent 
\small{ {\it Acknowledgements.}  The work of J.H. was supported by the Natural Sciences and Engineering Research Council of Canada (NSERC); that of A. Yu. O.  by RFBR grant 18-01-00273a and state assignment  N 0149-2019-0002. \hfill
\bigskip



\end{document}